\documentclass[journal]{IEEEtran}

\usepackage{amsmath,amssymb,amsfonts,mathtools}
\usepackage{bm}
\usepackage{amsthm}
\usepackage{algorithm}
\usepackage{algpseudocode}
\usepackage{graphicx}
\usepackage{booktabs}
\usepackage{multirow}
\usepackage{cite}
\usepackage{url}
\usepackage{xcolor}
\usepackage{siunitx}
\usepackage{hyperref}

\newtheorem{theorem}{Theorem}
\newtheorem{lemma}{Lemma}

\newtheorem{definition}{Definition}

\newtheorem{remark}{Remark}

\title{Information-Theoretic Limits of Integrated Sensing and Communication with Finite Learning Capacity}

\makeatletter
\def\blfootnote{\xdef\@thefnmark{}\@footnotetext}
\makeatother
\begin{document}
		\author{Farshad~Rostami~Ghadi,~\IEEEmembership{Member},~\textit{IEEE},~F.~Javier~L\'opez-Mart\'inez,~\IEEEmembership{Senior~Member},~\textit{IEEE},\\~Kai-Kit~Wong,~\IEEEmembership{Fellow},~\textit{IEEE},~and Christos~Masouros,~\IEEEmembership{Fellow},~\textit{IEEE}%,~and~Lajos~Hanzo,~\IEEEmembership{Life Fellow},~\textit{IEEE}
		}
		
	\maketitle
	
	\blfootnote{The work of F. Rostami Ghadi is supported by the European Union's Horizon 2022 Research and Innovation Programme under Marie Skłodowska-Curie Grant No. 101107993.}
	\blfootnote{The work of F. J. L\'opez-Mart\'inez is supported by grant PID2023-149975OB-I00 (COSTUME) funded by MICIU/AEI/10.13039/501100011033, and by ERDF/EU.}
	\blfootnote{The work of  K. K. Wong is supported by the Engineering and Physical Sciences Research Council (EPSRC) under Grant EP/W026813/1.}
	
%	\blfootnote{The work of L. Hanzo is supported by the Engineering and Physical Sciences Research Council (EPSRC) projects Platform for Driving Ultimate Connectivity (TITAN) under Grant EP/X04047X/1 and Grant EP/Y037243/1.}
	
	\blfootnote{\noindent F. Rostami Ghadi and F. J. L\'opez-Mart\'inez are with the Department of Signal Theory, Networking and Communications, Research Centre for Information and Communication Technologies (CITIC-UGR), University of Granada, 18071, Granada, Spain. (e-mail: $\rm\{ f.rostami,fjlm\}@ugr.es$).}

%\blfootnote{\noindent F. Rostami Ghadi is with the Department of Electronic and Electrical Engineering, University College London, WC1E 7JE London, U.K. (e-mail: $\rm f.rostamighadi@ucl.ac.uk$).}

%	\blfootnote{\noindent F. J. L\'opez-Mart\'inez is with the Department of Signal Theory, Networking and Communications, Research Centre for Information and Communication Technologies (CITIC-UGR), University of Granada, 18071, Granada, Spain, and also with the Communications and Signal Processing Lab, Telecommunication Research Institute (TELMA), Universidad de M\'alaga, M\'alaga, 29010, Spain. (e-mail: $\rm fjlm@ugr.es$).}
	
		\blfootnote{\noindent K. K. Wong is with the Department of Electronic and Electrical Engineering, University College London, WC1E 7JE London, U.K., and also with Yonsei Frontier Lab, Yonsei University,  Seoul, Korea. (e-mail: $\rm kai\text{-}kit.wong@ucl.ac.uk$).}

	\blfootnote{\noindent  C. Masouros
		is with the Department of Electronic and Electrical Engineering, University College London, London, U.K. (e-mail: $\rm c.masouros@ucl.ac.uk$).}
	
%		\blfootnote{\noindent L. Hanzo is with the School of Electronics and Computer Science, University of Southampton, Southampton, U.K. (e-mail: $\rm lh@ecs.soton.ac.uk$).}
	
	\blfootnote{Corresponding author: Farshad Rostami Ghadi.}
	\begin{abstract}
This paper develops a unified information-theoretic framework for artificial-intelligence (AI)-aided integrated sensing and communication (ISAC), where a learning component with limited representational capacity is embedded within the transceiver loop. The study introduces the concept of an AI capacity budget to quantify how the finite ability of a learning model constrains joint communication and sensing performance. Under this framework, the paper derives both converse (upper) and achievability (lower) bounds that define the achievable rate-sensing region. For Gaussian channels, the effect of limited learning capacity is shown to behave as an equivalent additive noise, allowing simple analytical expressions for the resulting communication rate and sensing distortion. The theory is then extended to Rayleigh and Rician fading as well as to multiple-input multiple-output (MIMO) systems through new matrix inequalities and a constructive mapping between AI capacity and effective noise covariance. Resource allocation between sensing and communication is optimized under this learning constraint, yielding closed-form conditions in the Gaussian case. A general learning-information trade-off law is also established, linking the representational power of the learning module to the achievable performance frontier. Finally, a practical variational training procedure is proposed to enforce the capacity constraint and to guide empirical evaluation. The derived scaling laws provide quantitative insight for co-designing model size, waveform, and hardware in next-generation ISAC systems.

	\end{abstract}
	
	\section{Introduction}
%	Integrated Sensing and Communication (ISAC) has emerged as a cornerstone technology for sixth-generation (6G) wireless networks, enabling the same waveform, spectrum, and hardware resources to be jointly used for both data communication and environmental sensing. By merging these traditionally separate functionalities, ISAC promises significant gains in spectral efficiency, latency reduction, and situational awareness.  
%	
%	Classical information-theoretic studies of ISAC focus on characterizing rate–distortion or capacity–sensing trade-offs under ideal algorithmic receivers that are assumed to process signals optimally. However, practical ISAC systems increasingly depend on \emph{learned components}—such as neural encoders, decoders, beamformers, and feature extractors—whose representational and computational capacities are inherently finite. These limitations, arising from quantization, compression, or model size, directly influence how much information can be preserved and exploited within the transceiver loop. Despite the proliferation of AI-aided ISAC prototypes and demonstrations, a rigorous theoretical framework that explicitly accounts for the learning capacity of such modules has not yet been established.

Integrated sensing and communication (ISAC) has become a key technology for sixth-generation (6G) wireless networks, where the same spectrum, waveform, and hardware resources are jointly used for data transmission and environmental perception \cite{zhang2021over}. By combining communication and sensing within a common physical layer, ISAC systems can achieve higher spectral and energy efficiency, lower latency, and improved situational awareness. This integration also enables new applications such as radar-assisted communication, vehicle perception, and joint localization and connectivity in intelligent networks \cite{liu2023seventy}.

Classical information-theoretic studies of ISAC have focused on fundamental trade-offs between communication rate and sensing accuracy. %\cite{xiong2023fund}. 
 Recent works have established formal ISAC information-theoretic models, including the capacity-distortion trade-off for memoryless ISAC channels \cite{AhmadipourJoint2024}, collaborative ISAC for multi-terminal systems \cite{AhmadipourCollaborative2023}, finite-blocklength ISAC bounds \cite{NikbakhtFinite2024}, and joint communication-state sensing under logarithmic-loss distortion \cite{JoudehCaire2024}. These works provide accurate rate-distortion characterizations under ideal transceiver assumptions. Such models provide important theoretical insights, but they assume optimal signal processing and unlimited representational capability at the receiver.
%These formulations often rely on idealized transceiver models that assume optimal signal processing and inference. While such models provide important theoretical insights, they do not capture the limitations of modern machine learning (ML)-based transceivers that are now being adopted in practice \cite{demirhan2023integrate}. 
%In practical ISAC architectures, many transceiver components, including modulators, channel estimators, beamformers, and target detectors, are implemented using learned models such as neural networks or data-driven feature extractors \cite{liu2025integrated, luong2025advanced}. 
In emerging ISAC architectures, several transceiver components, including channel estimation, beamforming, and target classification modules, are increasingly implemented using learning-based models, as demonstrated by recent experimental prototypes and automotive sensing systems. These models have finite representational and computational capacity determined by their size, quantization precision, and available training data. Such constraints create a learning bottleneck that limits how much information can be preserved and exploited throughout the transceiver chain \cite{Alemi}. The resulting system behavior differs from the assumptions of perfect information preservation that underlie classical ISAC theory, and the overall performance becomes jointly determined by both the physical channel and the learning capacity of the transceiver modules.

Although some artificial intelligence (AI)-aided ISAC prototypes and experimental platforms have been demonstrated \cite{vaezi2025ai, ald20256g}, these studies are primarily algorithmic and do not offer an analytical connection between learning capacity and fundamental ISAC limits. %there is still no accurate theoretical framework that explicitly incorporates the finite learning capacity of neural modules into the joint sensing and communication analysis.
 Establishing such a framework is essential for understanding how learning constraints influence achievable rates and sensing accuracy, and for guiding the design of future 6G systems that integrate AI as part of the physical layer.
\subsection{State-of-the-Art}
Information-theoretic foundations of ISAC have been extensively investigated over the past few years.  In particular, \cite{AhmadipourJoint2024} characterizes the exact capacity-distortion region for single-receiver ISAC, while \cite{AhmadipourCollaborative2023} extends these results to multiple access channel and device-to-device settings with collaborative state information. Furthermore, finite-blocklength behavior has been analyzed in \cite{NikbakhtFinite2024}, and log-loss sensing mutual information (MI) was developed in \cite{JoudehCaire2024}. These works establish the fundamental information-theoretic structure of ISAC, however, all assume ideal transceivers without representational constraints.

%\cite{liu2022survey} provided a survey on the fundamental limits of ISAC, characterizing rate-distortion and capacity-sensing trade-offs under ideal transceiver assumptions. \cite{xiong2023fund} analyzed the fundamental rate-Cram\'er-Rao bound (CRB) region of Gaussian ISAC channels, while  \cite{liu2022cro} optimized joint radar-communication beamforming using a CRB-based formulation under communication constraints. 

The extension of ISAC to multi-antenna and multi-user scenarios has also been a major research focus. The rate-Cram\'er-Rao bound (CRB) trade-off for multiple-input multiple-output (MIMO) ISAC and the optimal transmit covariance structures that balance communication throughput and sensing accuracy were investigated in \cite{hua2023mimo}. Also,\cite{an2023fund} examined the trade-off between detection probability and achievable rate in dual-functional radar-communication systems, while \cite{ren2024fund} extended this analysis to multi-target and multicast ISAC networks and characterized the achievable CRB-rate region for multi-antenna configurations. Although these works enhance the understanding of ISAC design across practical deployment scenarios, they generally assume ideal signal processing blocks and do not incorporate learning constraints within the transceiver.
 
 %These studies have built the theoretical foundation of ISAC performance limits, but they generally assume ideal receiver processing and do not account for algorithmic or learning-related limitations.
 
 %While these works enhance understanding of ISAC design across practical deployment scenarios, they continue to treat signal processing modules as ideal and do not explicitly model the internal capacity constraints associated with learning-based transceiver components.
	
Beyond theoretical formulations, optimization-driven approaches have been developed for dual-functional radar-communication (DFRC) systems. For instance, \cite{dong2023joint} proposed a secure joint beamforming and power allocation scheme for MIMO DFRC systems to satisfy both sensing and secrecy constraints. Furthermore, \cite{liu2020radar} introduced a radar-assisted predictive beamforming framework that exploits sensing feedback to improve communication reliability in vehicular networks. These methods focus on practical algorithmic strategies but remain deterministic and do not connect model complexity or neural-network capacity to fundamental ISAC limits.

%Although these methods address practical algorithmic designs, they remain deterministic and do not connect model complexity or neural network capacity to fundamental information-theoretic limits.

AI-enabled and learning-based ISAC systems have recently gained significant attention: \cite{zhang2025int} presented comprehensive surveys on intelligent ISAC architectures that incorporate deep learning (DL) for waveform design, feature extraction, and target recognition; \cite{ald20256g} provided an extensive review of ISAC technologies for 6G, discussing enabling methods, standardization progress, and the growing influence of AI and machine learning (ML) in both sensing-centric and communication-centric systems. Also, \cite{liu2024ai} proposed an AI-driven ISAC framework based on a federated fog network architecture and introduced learning-based interference management and mobility-aware control mechanisms. Recently, \cite{vaezi2025ai} further explored AI-empowered ISAC, demonstrating the use of DL for unified waveform and beamforming design to jointly optimize sensing and communication performance. However, while these works demonstrate the practical advantages of AI in ISAC, they do not establish information-theoretic performance limits under finite learning capacity. 

A related line of research examines representation-constrained learning and task-oriented communication: \cite{Shlezinger2021} proposed deep task-based quantization to optimize front-end quantizers for ML-driven receivers, while \cite{xi2021bit} studied bit-limited MIMO radar systems under task-based quantization principles. In a seminal paper,\cite{Tishby2000} introduced the information bottleneck framework, and \cite{Alemi} extended it to deep neural networks. These studies highlight the role of information constraints in learning, but none connect a finite MI representation budget to joint ISAC performance, nor do they provide an analytical capacity-distortion region under learning constraints.

	\subsection{Motivation and Contributions}
	Despite extensive progress in ISAC for multi-antenna design, DFRC optimization, and AI-enabled implementations, existing works either assume idealized receiver processing or focus on algorithmic designs without quantifying the effect of limited learning capacity.
	In contrast, modern AI-driven transceivers rely on finite models that can only capture a subset of the relevant signal statistics. This gap between ideal theoretical assumptions and practical learning-based implementations motivates a new perspective that incorporates learning constraints into information-theoretic modeling of ISAC systems. By incorporating this perspective, we can better understand how constraints on the learning process translate into measurable losses in sensing accuracy and communication rate, and how to co-design model capacity alongside waveform and power resources.
	
	Motivated by filling out the aforesaid gap, this work introduces an explicit \emph{AI capacity constraint} into the ISAC framework and develops the first unified information-theoretic characterization of its impact. The key novelty lies in modeling the learning component as a \emph{stochastic bottleneck} that is limited by a MI budget, representing the finite ability of the learning model to capture and transmit relevant features. This abstraction provides a clean analytical bridge between DL principles and classical communication theory. The main contributions are summarized as follows:
	\begin{itemize}
		\item \textbf{AI-ISAC capacity region:} We formalize the achievable rate-sensing region for ISAC systems that include a finite-capacity learning module. The model introduces an information bottleneck between transmitted and processed signals, defining a joint region of feasible communication rates and sensing distortions under a specified learning capacity.
		
		\item \textbf{Tight bounds and scaling laws:} We derive both converse (upper) and achievability (lower) bounds for the proposed framework. In the Gaussian case, limited learning capacity behaves as an equivalent additive noise, leading to an analytically tractable relation between model capacity and system performance. The resulting performance loss decays exponentially with increasing learning capacity, following a clear and interpretable scaling law.
		
		\item \textbf{Extension to fading and multiple-antenna channels:} The framework is generalized to Rayleigh and Rician fading environments through integral forms and tight bounds, and further extended to MIMO systems via a new matrix-based noise-allocation lemma that links learning capacity to an effective covariance structure.
		
		\item \textbf{Resource optimization under learning constraints:} We formulate and solve joint power and time allocation problems for sensing and communication tasks when the transceiver includes a finite-capacity learning component. For Gaussian channels, the Karush-Kuhn-Tucker (KKT) conditions yield closed-form solutions that highlight the coupling between physical resources and learning capacity.
		
		\item \textbf{Practical realization and training procedure:} A practical variational training algorithm is proposed to enforce the capacity constraint during network optimization using a differentiable MI penalty. This enables empirical validation of the theoretical framework and provides a bridge between information theory and modern DL practice.
	\end{itemize}
	
%	\subsection{Novelty}
%	The novelty of this work lies in treating \emph{AI model capacity as a fundamental physical resource} within ISAC. Unlike prior studies that either assume ideal receivers or focus purely on algorithmic design, this framework provides a quantitative theory connecting the limits of learning, communication, and sensing in a unified form. By deriving explicit capacity regions, scaling laws, and optimization results, the proposed model offers both theoretical insights and practical guidance for designing future AI-driven ISAC systems.

%	\begin{remark}
%		Throughout, rates are in bits per (complex) channel use; thus $I=\log_2(1+\mathrm{SNR})$ for complex AWGN. This fixes the dimensional constant often source of $1/2$ factors in real-valued channels.
%	\end{remark}
	
	\section{System Model and Preliminaries}
	% ================================
	
	In this section, we introduce the mathematical model for the considered AI-aided ISAC system and outline the basic definitions and performance metrics used throughout the paper. We begin with the baseband signal model, followed by the formulation of the learning module and its associated capacity constraint, and finally define the measures used to evaluate joint sensing and communication performance.
	
	\subsection{ISAC Baseband Model}
	%We consider a transceiver that performs both communication and sensing using a common waveform and shared resources. 
	As shown in Fig. \ref{fig:model1}, we consider a standard ISAC architecture consisting of a single transmitter equipped with a unified signaling module, a dedicated communication receiver, and a sensing receiver that observes target-dependent echoes. The transmitter emits a joint waveform that simultaneously conveys information to the communication receiver and probes the environment for sensing.
	The transmitted signal at time index $t$ is given by
	\begin{equation}
		x_t = s_t + c_t, \qquad \mathbb{E}[|x_t|^2] \le P,
	\end{equation}
	where $s_t$ and $c_t$ denote the sensing and communication components, respectively, and $P$ represents the total transmit power constraint. 
	
	At the receiver side, two types of observations are obtained; one dedicated to data communication and the other to environment sensing, i.e., 
	\begin{align}
		y_{c,t} &= h_c x_t + n_{c,t}, \qquad n_{c,t} \sim \mathcal{CN}(0,N_c), \label{eq:cobs}\\
		y_{s,t} &= h_s(\boldsymbol{\theta}) x_t + n_{s,t}, \qquad n_{s,t} \sim \mathcal{CN}(0,N_s), \label{eq:sobs}
	\end{align}
	where $h_c$ and $h_s(\boldsymbol{\theta})$ denote the complex channel coefficients for communication and sensing, respectively. Besides, $N_c$ and $N_s$ are the noise variances at the communication and sensing receivers, respectively. The sensing channel depends on an unknown parameter vector $\boldsymbol{\theta}$ that represents environment-related quantities such as delay, Doppler, or angle of arrival (AoA). Thus, the receiver aims to decode the transmitted communication message while simultaneously estimating $\boldsymbol{\theta}$ from the reflected/received signal.
		\begin{figure}[t]
		\centering
		\includegraphics[width=.9\columnwidth]{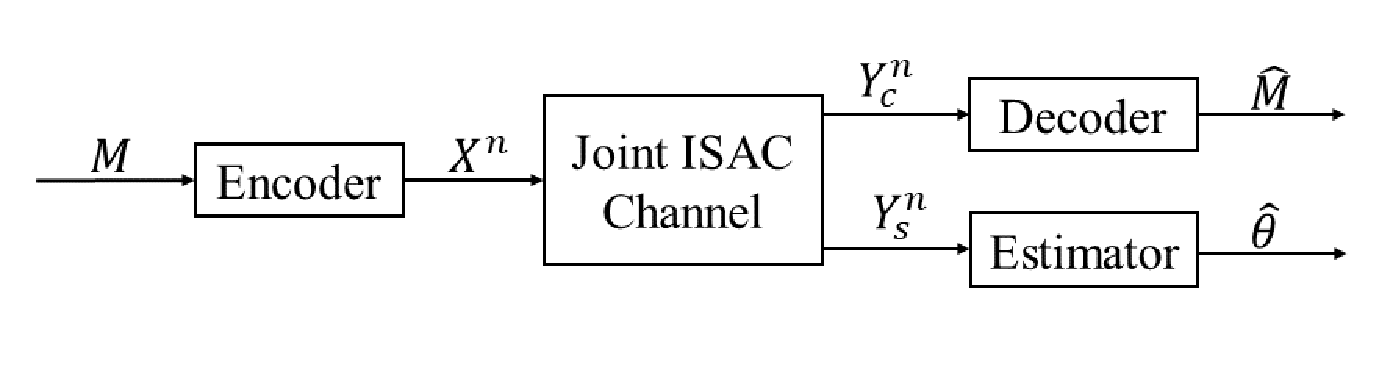}
		\caption{Physical baseband ISAC system model.}
		\label{fig:model1}
	\end{figure}
	\subsection{Learning Module and Capacity Budget}
	In conventional ISAC systems, both communication decoding and target estimation typically assume algorithms with unlimited representational ability. In contrast, the proposed \emph{learning-constrained} framework (see Fig. \ref{fig:model2}) explicitly models the finite capability of the embedded AI module. 
	
	The transceiver employs a learnable representation module $(\phi_e, \phi_d)$ that induces a latent variable $Z$ associated with the transmitted signal $X$.
	%The communication receiver is equipped with a learnable encoder-decoder pair $(\phi_e, \phi_d)$ that produces a latent representation $Z$ used for both message decoding and sensing inference. 
	 The sensing receiver relies on the same latent $Z$ together with its own physical observation $Y_s$. The encoder maps the input data sequence $X$ to a lower-dimensional latent representation as follow
	\begin{equation}
		Z = f_{\phi_e}(X),
	\end{equation}
	which serves as a compressed %|or abstracted 
	form of the transmitted signal. The decoder $g_{\phi_d}$ subsequently operates on the available observations together with this latent variable to perform message decoding or parameter estimation tasks. The variable $Z$ represents an internal learned representation associated with the transmitted signal $X$, capturing the finite modeling capacity of the AI-assisted transceiver, which is not a physical channel observation. We assume the latent representation $Z^n$ associated with 
	$X^n$ is available to the receiver as side information, e.g., via shared model state or embedded signaling.
	
	To capture the limited expressive power of the learning module, we impose an \emph{information capacity constraint} on the latent representation as
	\begin{equation}
		I(X;Z) \le C_{\mathrm{AI}}, \label{eq:budget}
	\end{equation}
	where $I(X;Z)$ denotes the MI between the input and the learned latent, and $C_{\mathrm{AI}}$ quantifies the maximum representational capacity of the AI model. In fact, the constraint in \eqref{eq:budget} is adapted from the information bottleneck principle, where $C_\mathrm{AI}$ quantifies the maximum number of relevant bits the learning module can extract from the transmitted signal. This capacity reflects representational limits arising from model size, quantization precision, and training complexity. In contrast to classical compression metrics,  $C_\mathrm{AI}$ characterizes the intrinsic expressiveness of the learned representation rather than a communication link budget. %Similar MI constraints have been used to model finite-representation learning systems in information theory and ML. 
	Similar information-constrained representations have been extensively used in learning theory, most notably through the information bottleneck framework \cite{Tishby2000} and its variational formulation for deep networks \cite{Alemi}.
	 %The resulting dependency structure forms the Markov chain
%	$
		%Z \leftarrow X \rightarrow (Y_c, Y_s)
%$, 
%	which formalizes the notion that all downstream inferences rely solely on the compressed representation $Z$ rather than the full input. 
The finite-capacity representation 
$Z$ limits how much structure of the transmitted signal can be exploited by downstream communication and sensing tasks.

Since the latent representation $Z$ is generated from the transmitted signal $X$ and does not influence the physical channel, the dependency structure follows the fork
$Z \leftarrow X \rightarrow (Y_c, Y_s)$, i.e., $Z - X - (Y_c, Y_s)$. 
Here, $Y_c$ and $Y_s$ denote the received communication and sensing signals, respectively, obtained through the corresponding communication and sensing channels.
			\begin{figure}[t]
		\centering
		\includegraphics[width=.9\columnwidth]{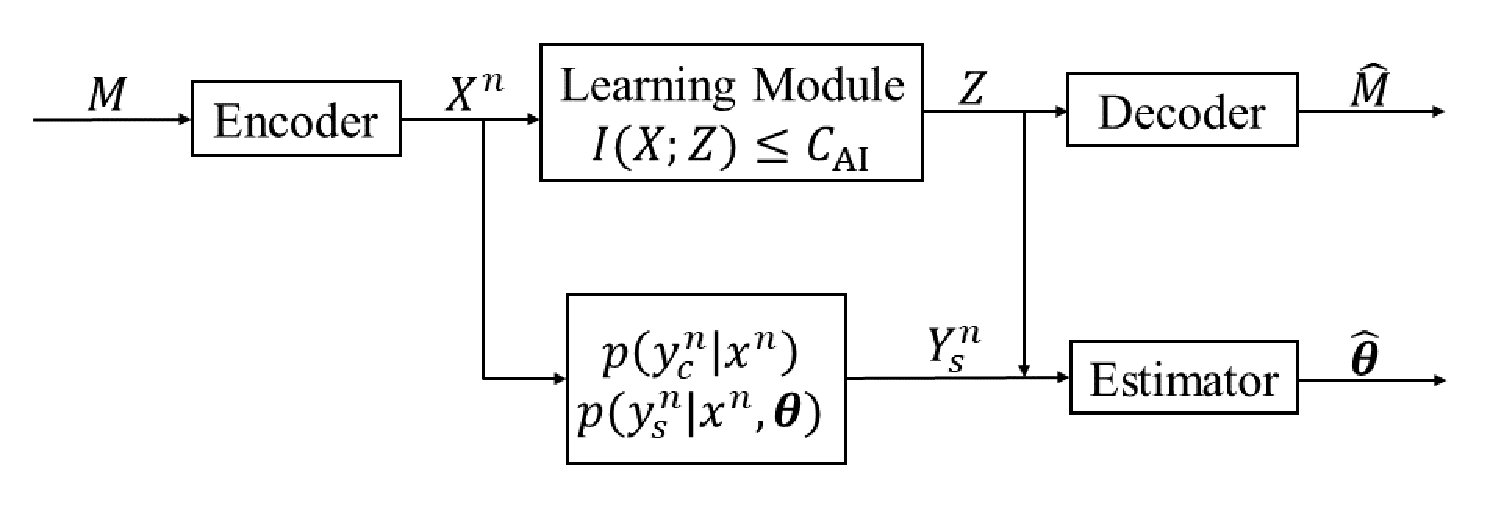}
		\caption{Information-theoretic abstraction of the learning-constrained ISAC system model.}
		\label{fig:model2}
	\end{figure}
	\subsection{Performance Metrics}
	The joint performance of communication and sensing is evaluated using standard information-theoretic and estimation-theoretic criteria \cite{AhmadipourJoint2024}.
	
	\subsubsection{Communication rate} The transmitter conveys a message $M \in \{1, 2, \dots, 2^{nR}\}$ over $n$ channel uses. A reliable scheme satisfies $\Pr\{\hat{M} \neq M\} \to 0$ as $n \to \infty$, where $R$ %(in bits per channel use) 
	represents the achievable communication rate\footnote{Throughout, rates are expressed in bits per (complex) channel use; thus $I=\log_2(1+\mathrm{SNR})$ for complex low-pass equivalent signals under AWGN.} %This fixes the dimensional constant often source of $1/2$ factors in real-valued channels.}.
	
	\subsubsection{Sensing distortion} The receiver estimates the unknown environmental parameter $\boldsymbol{\theta}$ based on the sensing observations and the latent representation $(Y_s, Z)$. The accuracy of this estimate $\hat{\boldsymbol{\theta}}$ is measured by the expected sensing distortion as follow
	\begin{equation}
		D_s = \mathbb{E}\!\left[ d(\boldsymbol{\theta}, \hat{\boldsymbol{\theta}}) \right],
	\end{equation}
	where $d(\cdot,\cdot)$ is a suitable distance or error measure, such as normalized mean-squared error (NMSE). For parametric estimation problems, we also use the Bayesian Fisher information $\mathcal{I}(\boldsymbol{\theta})$ and the corresponding Cram\'er-Rao lower bound (CRLB) as analytical surrogates.
	
	Together, the pair $(R, D_s)$ characterizes the joint communication-sensing performance under a given AI capacity constraint $C_{\mathrm{AI}}$. 
	
	\section{Learning-Constrained Information-Theoretic Region}
	
	This section establishes the fundamental limits of ISAC systems when the transceiver employs a finite-capacity learning-based representation of the transmitted signal. We first define the achievable region under a learning-capacity constraint, then provide general converse and achievability results that bound this region from above and below. Finally, we present the sensing information map for Gaussian observations and state a general learning-information trade-off law linking model generalization to achievable performance.
	
	\subsection{Definition of the AI-ISAC Region}
	To formalize the performance limits of learning-constrained ISAC, we characterize the set of all achievable communication rates and sensing distortions under a given learning-capacity budget. 
	
	\begin{definition}[AI-ISAC Region]
		The achievable region associated with a learning capacity constraint $C_{\mathrm{AI}}$ is defined as
		\begin{equation}
			\mathcal{R}_{\mathrm{AI\!-\!ISAC}}(C_{\mathrm{AI}})=
			\left\{(R,D_s):
			\begin{array}{l}
				R \le I(X;Y_c|Z),\\[2pt]
				D_s \ge f\big(I(X;Y_s|Z)\big),\\[2pt]
				I(X;Z) \le C_{\mathrm{AI}}
			\end{array}
			\right\},
			\label{eq:region}
		\end{equation}
		where $f(\cdot)$ denotes the information-to-distortion mapping associated with the sensing task. This region quantifies the joint communication and sensing performance attainable when the learning module can represent at most $C_{\mathrm{AI}}$ bits of information about the transmitted signal $X$.
	\end{definition}
	
	\subsection{Converse and Achievability Bounds}
	The following theorems provide upper and lower bounds on the achievable rate-distortion pair $(R, D_s)$ within the defined region.
	
	\begin{theorem}[Converse]\label{th:converse}
		For any system satisfying the learning capacity constraint \eqref{eq:budget}, the achievable pair $(R, D_s)$ is bounded as
		\begin{align}
			R &\le I(X;Y_c,Y_s) - \Delta_R(C_{\mathrm{AI}}), \label{eq:rbub}\\
			D_s &\ge f\!\big(I(X;Y_c,Y_s) - \Delta_D(C_{\mathrm{AI}})\big), \label{eq:dsbub}
		\end{align}
		where $\Delta_R(C_{\mathrm{AI}})$ and $\Delta_D(C_{\mathrm{AI}})$ are non-negative functions that decrease monotonically with $C_{\mathrm{AI}}$ and vanish as $C_{\mathrm{AI}}\!\to\!\infty$.
	\end{theorem}
	
	\begin{proof}
		From the data processing inequality, $I(X;Y_c,Y_s) \ge I(Z;Y_c,Y_s)$ since $Z$ is a function of $X$. Because the learning module is limited by $I(X;Z)\le C_{\mathrm{AI}}$, let $\Delta = I(X;Y_c,Y_s) - I(Z;Y_c,Y_s) \ge 0$. Combining this with the standard converse bound based on Fano's inequality gives~\eqref{eq:rbub}. The sensing inequality~\eqref{eq:dsbub} follows from applying the same gap $\Delta$ within the information-to-distortion function $f(\cdot)$. Detailed derivations are provided in Appendix~\ref{app:converse}.
	\end{proof}
	
	\begin{theorem}[Achievability]\label{th:ach}
		There exists a stochastic encoder $p(z|x)$ satisfying $I(X;Z)\le C_{\mathrm{AI}}$ and corresponding decoders such that
		\begin{align}
			R &\ge I(X;Y_c|Z) - \epsilon_R, \\
			D_s &\le f^{-1}\!\big(I(X;Y_s|Z)\big) + \epsilon_D,
		\end{align}
		where $\epsilon_R, \epsilon_D \to 0$ as the blocklength $n \to \infty$.
	\end{theorem}
	
	\begin{proof}
		The proof follows standard random coding arguments with an auxiliary random variable $Z$ satisfying the given MI constraint. Separate decoders for communication and sensing operate conditionally on $Z$. Rate-distortion achievability results then yield the desired performance bounds. Details are provided in Appendix~\ref{app:ach}.
	\end{proof}
	
	Theorems~\ref{th:converse} and~\ref{th:ach} together define an achievable region that shrinks as the available learning capacity $C_{\mathrm{AI}}$ decreases and approaches the classical ISAC limits as $C_{\mathrm{AI}}\!\to\!\infty$.
	
	\subsection{Sensing Information Map}\label{subsec:rdmap}
	To relate the information available for sensing to the corresponding estimation distortion, we introduce a sensing rate-distortion function. For the canonical case of minimum mean-squared error (MMSE) estimation over a complex Gaussian channel, we have
	\begin{align}
		D_s = \frac{\sigma_\theta^2}{1+\mathrm{SNR}_s^{\mathrm{eff}}},
		%\quad\Longleftrightarrow\quad
	%	f(I) = \frac{\sigma_\theta^2}{1+(2^{I}-1)}, \label{eq:gaussRD}
	\end{align}
	and using $I=\log_2\left(1+\mathrm{SNR}_s^\mathrm{eff}\right)$, 
	
	\begin{align}
f(I) = \frac{\sigma_\theta^2}{1+(2^{I}-1)}, \label{eq:gaussRD}
	\end{align}
	where $\mathrm{SNR}_s^{\mathrm{eff}}$ denotes the effective sensing signal-to-noise ratio (SNR) at the estimator output and $\sigma^2_\theta$ denotes the variance of the sensing parameter $\theta$. Besides, the function $f\left(\cdot\right)$ follows from the Gaussian MMSE sensing model by expressing the effective sensing SNR in terms of MI. This relation establishes a direct mapping between the MI $I(X;Y_s|Z)$ and the achievable sensing distortion $D_s$.
	
	\subsection{Learning-Information Trade-Off Law}
	The inclusion of a finite-capacity learning module introduces an additional dimension to the performance trade-off, linking achievable rate-distortion pairs to the generalization ability of the learning model itself.
	
	\begin{theorem}[Learning-Information Trade-Off]\label{th:lit}
		Let $\epsilon_{\mathrm{gen}}$ denote the generalization error of the learning module trained on $n_{\mathrm{tr}}$ samples. Under the MI constraint \eqref{eq:budget} and mild regularity conditions (see Appendix~\ref{app:gen}), there exist positive constants $\delta$ and $\beta$ such that
		\begin{equation}
			R + \delta D_s \le I(X;Y_c,Y_s) - \beta\, \epsilon_{\mathrm{gen}}.
			\label{eq:lit}
		\end{equation}
	\end{theorem}
	
	\begin{proof}
		The result follows by combining information-stability bounds on the generalization error with the converse in Theorem~\ref{th:converse}. The term $\epsilon_{\mathrm{gen}}$ effectively reduces the MI available for both communication and sensing tasks, tightening the achievable frontier. Full derivations are presented in Appendix~\ref{app:gen}.
	\end{proof}
	
	The above theorem provides an intuitive interpretation that learning capacity and generalization directly limit how closely a practical AI-aided ISAC system can approach the ideal information-theoretic bounds. As $C_{\mathrm{AI}}$ and the quality of training increase, the gap to the optimal rate-distortion frontier narrows accordingly.
	
\section{Closed-Form Analysis: Gaussian, Rayleigh, Rician, and MIMO}\label{sec-perform}

This section provides explicit analytical expressions for the achievable performance of the proposed AI-ISAC framework under several representative channel models. We begin with the complex Gaussian case, which offers the cleanest closed-forms and reveals the fundamental role of the AI capacity parameter. The results are then extended to Rayleigh and Rician fading environments, followed by a generalization to MIMO systems through a covariance mapping lemma.

\subsection{Complex Gaussian ISAC}\label{sec:gauss}
We consider the baseband model in \eqref{eq:cobs} and \eqref{eq:sobs} with $X\sim\mathcal{CN}(0,P)$. The finite-capacity learning module can be modeled through a Gaussian auxiliary channel, that compresses $X$ into a latent representation, i.e.,
\begin{equation}
	Z = X + W_z, \qquad W_z \sim \mathcal{CN}(0, N_z), \quad W_z \perp X,
	\label{eq:zmodel}
\end{equation}
where $W_z$ represents the effective AI noise induced by the limited learning capacity. The corresponding MI satisfies
\begin{equation}
	I(X;Z) = \log_2\!\left(1+\frac{P}{N_z}\right) \le C_{\mathrm{AI}},
\end{equation}
which leads to the equivalent noise variance as
\begin{equation}
	N_z = \frac{P}{2^{C_{\mathrm{AI}}}-1}.
	\label{eq:Nz}
\end{equation}

\subsubsection{Effective SNRs}  
The capacity limitation manifests as a degradation in both communication and sensing SNRs, i.e., 
\begin{align}
	\tilde{\gamma}_c &= \frac{|h_c|^2 P}{N_c + |h_c|^2 N_z}, \qquad
	\tilde{\gamma}_s = \frac{|h_s|^2 P}{N_s + |h_s|^2 N_z}.
	\label{eq:effsnr}
\end{align}

\subsubsection{Closed-form performance}  
The achievable communication rate and the resulting sensing distortion under the AI constraint are given by
\begin{align}
	R_{\mathrm{AI}} &= \log_2\!\left(1+\tilde{\gamma}_c\right), \label{eq:rateG}
	\end{align}
	and
	\begin{align}
	D_{s,\mathrm{AI}} &= \frac{\sigma_\theta^2}{1+\tilde{\gamma}_s}. \label{eq:distG}
\end{align}

\begin{remark}[Scaling Behavior]
	As $C_{\mathrm{AI}}\!\to\!\infty$, the artificial noise (AN) variance $N_z$ in \eqref{eq:Nz} tends to zero, and the system approaches the classical Shannon capacity and CRLB limits. Conversely, when $C_{\mathrm{AI}}\!\rightarrow\!0^+$, the virtual noise becomes dominant and both communication and sensing performance degrade sharply. The performance gap with respect to the ideal case scales approximately as $\mathcal{O}(2^{-C_{\mathrm{AI}}})$, providing a simple rule of thumb linking model capacity to achievable gain.
\end{remark}

\subsection{Rayleigh Block Fading}
To capture realistic wireless environments, we now consider Rayleigh block fading where $h_i \!\sim\! \mathcal{CN}(0,1)$ and define $\bar{\gamma}_i=P/N_i$. Conditioned on the fading magnitude $|h_i|^2=x$, the effective SNR becomes
\begin{align}
\tilde{\gamma}_i = \frac{x\,\bar{\gamma}_i}{1 + x\,\bar{\gamma}_i\,\kappa},
\qquad \kappa := \frac{N_z}{P} = \frac{1}{2^{C_{\mathrm{AI}}}-1}.
\end{align}
The corresponding ergodic rate is defined as
\begin{equation}
	\bar{R}_{\mathrm{AI}}
	= \frac{1}{\ln 2} \int_{0}^{\infty}
	\ln\!\left(1+\frac{x\,\bar{\gamma}_c}{1+x\,\bar{\gamma}_c\,\kappa}\right)
	e^{-x} dx.
	\label{eq:rayInt}
\end{equation}
Although the integral does not admit a simple closed-form, it can be tightly approximated using standard inequalities, e.g., Jensen bound, or computed numerically via Gauss-Laguerre quadrature as follows
\begin{equation}
	\bar{R}_{\mathrm{AI}}
	\approx \frac{1}{\ln 2} \sum_{m=1}^{M} w_m
	\ln\!\left(1+\frac{\xi_m \bar{\gamma}_c}{1+\xi_m \bar{\gamma}_c \kappa}\right),
\end{equation}
where $(\xi_m, w_m)$ denote the quadrature nodes and weights. The corresponding average sensing distortion is obtained analogously by replacing $\bar{\gamma}_c$ with $\bar{\gamma}_s$ in \eqref{eq:distG} and averaging over the fading distribution (see appendix~\ref{app:ray}).

\subsection{Rician Fading}
For environments with a deterministic line-of-sight (LoS) component, we model the channel as $h_i = \mu_i + \tilde{h}_i$, where $\tilde{h}_i\!\sim\!\mathcal{CN}(0,1)$ and the Rician factor is $K_i = |\mu_i|^2$. With $x=|h_i|^2$ following a non-central $\chi^2$ distribution, the ergodic rate can be written as
\begin{equation}
	\bar{R}_{\mathrm{AI}}^{\mathrm{Ric}}
	= \frac{1}{\ln 2} \int_{0}^{\infty}
	\ln\!\left(1+\frac{x\,\bar{\gamma}_c}{1+x\,\bar{\gamma}_c\,\kappa}\right)
	f_{|h|^2}^{\mathrm{Ric}}(x;K)\, dx,
\end{equation}
where $f_{|h|^2}^{\mathrm{Ric}}(x;K)$ denotes the non-central chi-square probability density function (PDF). A compact and accurate approximation can be obtained through moment matching, using $\mathbb{E}[x]=1+K$, which yields
\begin{equation}
	\bar{R}_{\mathrm{AI}}^{\mathrm{Ric}}
	\approx \log_2\!\left(1+\frac{(1+K)\bar{\gamma}_c}
	{1+(1+K)\bar{\gamma}_c\,\kappa}\right).
\end{equation}
This approximation provides excellent accuracy for moderate $K$ values, while more precise bounds can be derived via the Marcum-$Q$ representation (see Appendix~\ref{app:ric}).

\subsection{MIMO ISAC}
We now consider a MIMO ISAC system with $N_t$ transmit antennas and $N_r$ receive antennas. Therefore, the received communication and sensing signals $\mathbf{Y}_c, \mathbf{Y}_s \in \mathbb{C}^{N_r \times n}$ with $n$ denoting the blocklength are defined respectively as 
\begin{align}
	\mathbf{Y}_c = \mathbf{H}_c \mathbf{X} + \mathbf{N}_c, \qquad
	\mathbf{Y}_s = \mathbf{H}_s \mathbf{X} + \mathbf{N}_s,
\end{align}
where $\mathbf{X} \in \mathbb{C}^{N_t \times n}$ denotes the transmitted signal matrix,  $\mathbf{N}_c, \mathbf{N}_s \in \mathbb{C}^{N_r \times n}$ denote the noise matrices including i.i.d. $\mathcal{CN}(0,N_0)$ entries, and 
$\mathbf{H}_c, \mathbf{H}_s \in \mathbb{C}^{N_r \times N_t}$ are the communication and sensing channel matrices. The limited learning capacity is modeled by an \emph{AI-noise covariance} $\mathbf{R}_z \succeq 0$ that satisfies a capacity-budget mapping described in Lemma~\ref{lem:covmap}. The achievable communication rate then takes the form
\begin{equation}
	R_{\mathrm{AI}}^{\mathrm{MIMO}}
	= \log_2 \det\!\Big(
	\mathbf{I}_{N_r} +
	\mathbf{H}_c \mathbf{Q} \mathbf{H}_c^{H}
	(\mathbf{R}_c + \mathbf{H}_c \mathbf{R}_z \mathbf{H}_c^{H})^{-1}
	\Big),
	\label{eq:mimoRate}
\end{equation}
where  $\mathbf{Q}=\mathbb{E}[\mathbf{X}\mathbf{X}^H]$ the transmit covariance matrix. 
For the sensing function, the Fisher information of a linear Gaussian parameter model yields
\begin{equation}
	\mathcal{I}_\theta
	= \frac{\partial \boldsymbol{\mu}^H}{\partial \theta}
	\Big(\mathbf{R}_s + \mathbf{H}_s \mathbf{R}_z \mathbf{H}_s^H\Big)^{-1}
	\frac{\partial \boldsymbol{\mu}}{\partial \theta},
	\,\,
	\mathrm{CRLB}(\theta) \ge \mathcal{I}_\theta^{-1}.
\end{equation}

\begin{lemma}[Covariance Mapping]\label{lem:covmap}
	Let $\mathbf{X}\!\sim\!\mathcal{CN}(\mathbf{0},\mathbf{Q})$ and $\mathbf{Z}=\mathbf{X}+\mathbf{W}$ with $\mathbf{W}\!\sim\!\mathcal{CN}(\mathbf{0},\mathbf{R}_z)$ independent of $\mathbf{X}$. Then
	\begin{equation}
		I(\mathbf{X};\mathbf{Z})
		= \log_2 \frac{\det(\mathbf{Q}+\mathbf{R}_z)}{\det(\mathbf{R}_z)}
		\le C_{\mathrm{AI}}.
	\end{equation}
	A feasible $\mathbf{R}_z$ is any matrix satisfying
	$\det(\mathbf{I}+\mathbf{R}_z^{-1}\mathbf{Q}) \le 2^{C_{\mathrm{AI}}}$.
	The minimum-trace solution is obtained when $\mathbf{R}_z = \zeta \mathbf{Q}$ with 
	\[
	\zeta = (2^{C_{\mathrm{AI}}/r}-1)^{-1},
	\]
	on the $r=\mathrm{rank}(\mathbf{Q})$ active subspace, where $\zeta$ is a noise-scaling factor.
\end{lemma}

\begin{proof}
	The result follows from the closed-form expression for MI in Gaussian vector channels and the eigenvalue decomposition of $\mathbf{Q}$. Detailed derivations are provided in Appendix~\ref{app:mimo}.
\end{proof}
	
	\section{Resource Allocation under Learning Constraint}\label{sec-res}
	In this section, we investigate how the available transmission resources should be allocated between communication and sensing when the transceiver operates under a finite learning capacity. In practical ISAC systems, both power and time resources must be jointly managed, as the learning module introduces additional coupling between the two tasks through the capacity parameter $C_{\mathrm{AI}}$.
	
	\subsection{Problem Formulation}
	We consider a transmission frame with total power and duration $(P, T)$ that are divided between communication and sensing, i.e., 
	$
	P = P_c + P_s$ and $T = T_c + T_s$, respectively. 
	
	The resource allocation problem aims to balance the achievable communication rate and sensing accuracy according to a weighting parameter $\lambda \in [0,1]$. For the Gaussian ISAC model, the optimization problem can be formulated as
\begin{subequations}
	\begin{align}\label{eq:allocOpt}
	&	\max_{P_c,P_s,T_c,T_s}
		\quad
		\lambda\, R_{\mathrm{AI}}(P_c,T_c)
		- (1-\lambda)\, D_{s,\mathrm{AI}}(P_s,T_s)\\
	&\qquad	\text{s.t.} \qquad \quad P_c + P_s = P,\\
	&	\qquad \qquad \qquad T_c + T_s = T,\\
	& \qquad \qquad \qquad	I(X;Z) \le C_{\mathrm{AI}}.
	\end{align}
\end{subequations}
	The objective combines the communication benefit and sensing cost into a single scalar function that reflects the system’s operational priority. The last constraint ensures that the learning module does not exceed its information processing capacity.
	
	\subsection{KKT Conditions and Interpretation}
	Applying the KKT conditions to the problem in \eqref{eq:allocOpt} yields a pair of coupled optimality equations that can be interpreted as a form of \emph{learning-constrained} waterfilling. Specifically,
	\begin{align}
		\frac{\partial R_{\mathrm{AI}}}{\partial P_c}
		&= \eta
		+ \underbrace{\frac{\partial N_z}{\partial P_c}\cdot \Psi_c}_{\text{capacity coupling}},
		\label{eq:kktC}\\[2pt]
		\frac{\partial D_{s,\mathrm{AI}}}{\partial P_s}
		&= \eta'
		+ \underbrace{\frac{\partial N_z}{\partial P_s}\cdot \Psi_s}_{\text{capacity coupling}},
		\label{eq:kktS}
	\end{align}
	where $\eta$ and $\eta'$ are the Lagrange multipliers associated with the total power and time constraints, respectively, and $\Psi_c$, $\Psi_s$ represent the sensitivity of each task to changes in the effective AI noise power $N_z$. 
	
	Because the latent representation is itself affected by the transmit energy, the equivalent noise variance $N_z$ follows \eqref{eq:Nz} with $P$ replaced by the instantaneous energy assigned to the latent. This creates a non-trivial coupling between communication and sensing, which distinguishes the learning-constrained case from classical power allocation problems.
	
	\subsection{Closed-Form Solution for the Gaussian Case}
	For Gaussian ISAC channels, the system of equations in \eqref{eq:kktC} and \eqref{eq:kktS} admits an explicit analytical solution for $(P_c, P_s)$. The resulting expressions can be written in terms of the Lambert-$W$ function, which naturally appears when solving transcendental equations involving $P$ both in the numerator and denominator of logarithmic terms. Full derivations and detailed expressions are provided in Appendix~\ref{app:alloc}.
	
	The optimal solution demonstrates that, under tight learning capacity budgets, more power should be directed toward the component (sensing or communication) that exhibits higher information sensitivity with respect to $C_{\mathrm{AI}}$. As $C_{\mathrm{AI}}$ increases, the allocation gradually converges to the classical waterfilling solution, recovering standard ISAC resource trade-offs.

	\section{Deep Model Realization and Complexity}

	This section outlines a practical DL realization of the proposed learning-constrained ISAC framework and discusses its computational complexity. Our  objective here is to show how the theoretical MI constraint can be enforced in a differentiable form suitable for end-to-end training.
	
	\subsection{Variational MI-Constrained Training}
	Algorithm~\ref{alg:VIB} summarizes the learning process. To implement the information bottleneck defined by the constraint $I(X;Z)\!\le\!C_{\mathrm{AI}}$, we adopt a variational formulation based on a Gaussian latent representation. The encoder outputs a mean and covariance pair as
	\begin{align}
	q_\phi(z|x) = \mathcal{CN}\big(\mu_\phi(x), \Sigma_\phi(x)\big),
	\end{align}
	while the prior over latent variables is chosen as
$
	p(z) = \mathcal{CN}\big(0, \sigma_p^2 \mathbf{I}\big),
$
	serving as a reference distribution for the MI regularization.
	
	Instead of computing $I(X;Z)$ directly, which is generally intractable for high-dimensional data, we use the Kullback-Leibler (KL) divergence $D_{\mathrm{KL}}[q_\phi(z|x)\,\|\,p(z)]$ as a differentiable surrogate. The overall training objective combines the communication and sensing losses with this variational penalty:
	\begin{align}\notag
	&	\mathcal{L}(\phi,\psi) =
		\underbrace{\mathcal{L}_{\mathrm{comm}}(\psi)}_{\text{cross-entropy or MI-based loss}}
		+ \lambda\,
		\underbrace{\mathcal{L}_{\mathrm{sense}}(\psi)}_{\text{MSE or negative log-likelihood}}\\
		&
		+ \beta \Big(\mathbb{E}_{x}\, D_{\mathrm{KL}}\big[q_\phi(z|x)\|p(z)\big]
		- C_{\mathrm{AI}}\Big)_+,
		\label{eq:trainObjective}
	\end{align}
	where $\phi$ and $\psi$ denote the parameters of the encoder and decoder networks, respectively. The last term enforces the information-capacity budget through a soft penalty controlled by the coefficient $\beta$. The reparameterization trick enables low-variance gradient estimation by expressing $z = \mu_\phi(x) + \Sigma_\phi(x)^{1/2} \epsilon$ with $\epsilon \sim \mathcal{CN}(0, \mathbf{I})$.
	
	\begin{algorithm}[t]
		\caption{Variational Training for AI-ISAC under Information Constraint $I(X;Z)\le C_{\mathrm{AI}}$}
		\label{alg:VIB}
		\begin{algorithmic}[1]
			\State \textbf{Input:} Training dataset $\{(x_i, y_{c,i}, y_{s,i})\}$, target capacity $C_{\mathrm{AI}}$, and weights $\lambda, \beta$.
			\Repeat
			\State Sample a minibatch $\mathcal{B}$ and compute $\mu_\phi(x)$, $\Sigma_\phi(x)$.
			\State Reparameterize latent vectors: $z = \mu_\phi(x) + \Sigma_\phi(x)^{1/2}\epsilon$, with $\epsilon \sim \mathcal{CN}(0, \mathbf{I})$.
			\State Decode estimates: $\hat{m}, \hat{\boldsymbol{\theta}} = g_\psi(y_c, y_s, z)$.
			\State Evaluate total loss $\mathcal{L}$ from \eqref{eq:trainObjective}; backpropagate and update $(\phi, \psi)$.
			\Until{validation convergence}
			\State \textbf{Output:} Trained parameters $(\phi^\star, \psi^\star)$ satisfying $\mathbb{E}[D_{\mathrm{KL}}]\!\approx\!C_{\mathrm{AI}}$.
		\end{algorithmic}
	\end{algorithm}
	
	This training strategy effectively realizes the theoretical AI bottleneck in practice, allowing the system to learn latent representations whose information content is explicitly bounded. By adjusting $C_{\mathrm{AI}}$, one can control the trade-off between accuracy and efficiency in both sensing and communication tasks.
	
	\subsection{Computational Complexity}
	Let the encoder and decoder networks comprise $p$ trainable parameters with intermediate feature maps of width $w$. The per-iteration training complexity is dominated by the forward-backward propagation cost, which scales as $O(pw)$. Enforcing the KL-based MI regularization introduces an additional $O(d)$ computation per minibatch, where $d$ denotes the latent dimensionality of the representation $Z$. Since $d \ll pw$ in all considered architectures, the MI-penalty contributes only a negligible overhead relative to the overall training cost.
	
	The expression in \eqref{eq:Nz} reveals that each additional bit of learning capacity $C_{\mathrm{AI}}$ increases the allowable latent-space SNR by approximately a factor of two, reflecting the exponential sensitivity of the equivalent noise variance to the available information budget. Nevertheless, the empirical behavior of the system exhibits diminishing returns, namely, once the capacity exceeds roughly five to six bits per latent dimension, the corresponding gains in communication rate and sensing accuracy fall below one percent. This saturation effect provides a practical guideline for selecting both the model size and the desired learning-capacity allocation, as increasing $C_{\mathrm{AI}}$ beyond this regime yields minimal improvement while incurring higher computational cost.

	% ================================
	\section{Numerical Results}
	% ================================
In this section, we present numerical results to illustrate and verify the theoretical findings developed in the previous sections. The analytical expressions for the achievable communication rate and sensing distortion are evaluated under the considered Gaussian, Rayleigh, and Rician ISAC channel models. Unless otherwise stated, the simulation parameters are summarized in Table~\ref{tab:sim}. 
	\begin{table}[t]
		\centering
		\caption{Simulation parameters.}
		\begin{tabular}{lcc}
			\toprule
			Quantity & Symbol & Value \\
			\midrule
			Transmit power & $P$ & \SI{10}{dBm} \\
			Noise variances & $N_c,N_s$ & $0.1$ \\
			Carrier & $f_c$ & \SI{28}{GHz} \\
			AI capacity & $C_{\mathrm{AI}}$ & $0$--$8$ bits/use \\
			Blocklength & $n$ & $10^4$ \\
			MIMO & $(N_t,N_r)$ & $(2,2)$ or $(4,4)$ \\
			Latent dim & $d$ & $8$ or $16$ \\
			\bottomrule
		\end{tabular}\label{tab:sim}
	\end{table}

Fig. \ref{fig:r_c} illustrates the achievable communication rate as a function of the available AI capacity $C_{\mathrm{AI}}$ under three representative fading environments: Gaussian (AWGN), Rayleigh (ergodic), and Rician fading with $K = 6$ dB. We can see that as $C_{\mathrm{AI}}$ increases, the rate grows monotonically for all channel types.
%This behavior reflects the fact that a larger AI-bottleneck budget allows the transmitter-receiver pair (AI encoder-decoder) to exchange richer latent representations of the observed channel state, thereby reducing the effective uncertainty in the communication link. 
This behavior reflects the fact that a larger AI-bottleneck budget allows the transceiver to preserve a richer latent representation of the transmitted waveform (and task-relevant features) in $Z$, thereby reducing the effective information loss introduced by the learning bottleneck and improving the achievable communication throughput. It can be seen at low $C_{\mathrm{AI}}$, the AI module operates in a severely compressed regime, i.e., only coarse channel features are preserved, so the equivalent MI between the physical signal and its AI representation is small, resulting in limited throughput.
As $C_{\mathrm{AI}}$ rises, the AI representation becomes more expressive and the achievable rate asymptotically approaches the classical Shannon limit for each fading law.

	\begin{figure}[t]
		\centering
		\includegraphics[width=.9\columnwidth]{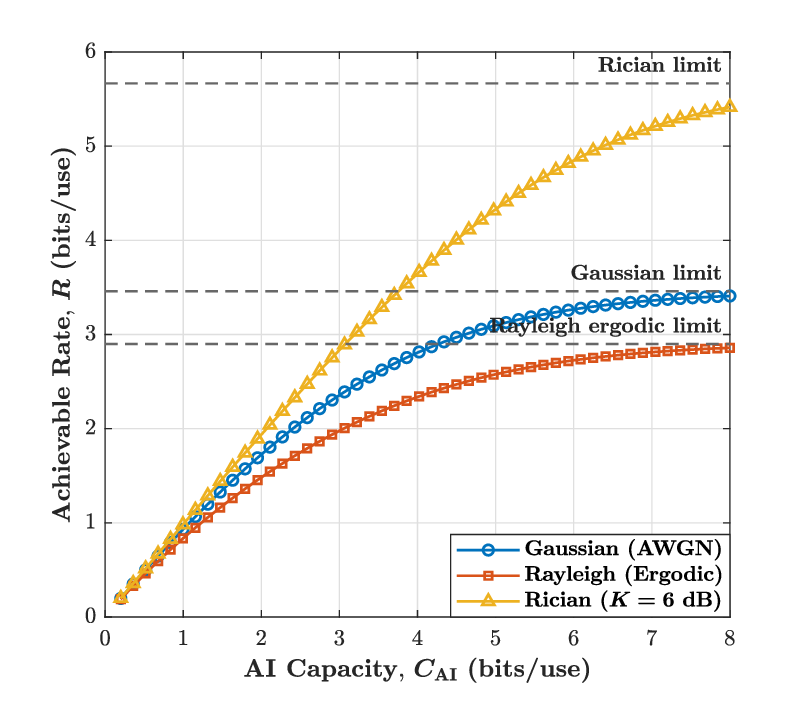}
		\caption{Achievable communication rate $R$ versus AI capacity $C_\mathrm{AI}$ under different fading channels.}
		\label{fig:r_c}
	\end{figure}

Furthermore, among the three curves, the Rician channel achieves the highest rate due to the presence of a deterministic LoS component that mitigates deep fading. The Gaussian curve corresponds to the ideal non-fading case and serves as an upper reference for ergodic fading environments, while the Rayleigh case yields the lowest rate due to strong channel amplitude fluctuations.
%Nevertheless, the relative gap between Rayleigh and Gaussian narrows as $C_{\mathrm{AI}}$ grows, confirming that the learning-constrained encoder progressively compensates for small-scale fading by leveraging the increased internal representation bandwidth. 
Nevertheless, the relative gap between Rayleigh and Gaussian narrows as $C_\mathrm{AI}$ grows, confirming that increasing the representational budget reduces the bottleneck-induced penalty and allows the receiver to exploit the physical observations more effectively despite small-scale fading. Therefore, we can generally find that Fig.\ref{fig:r_c} quantifies the direct benefit of expanding the AI capacity budget in joint communication-sensing transceivers even moderate increases in $C_{\mathrm{AI}}$ (from $1$ to $4$ bits/use) yield substantial rate gains, while larger capacities exhibit diminishing returns once the AI bottleneck ceases to dominate the system performance.

Fig. \ref{fig:s_c} depicts the sensing distortion $D_s$ as a function of the AI capacity $C_{\mathrm{AI}}$ for Gaussian, Rayleigh, and Rician fading environments. In all cases, $D_s$ decreases monotonically as $C_{\mathrm{AI}}$ increases, demonstrating that enlarging the information bottleneck of the AI module enables a more accurate reconstruction of the sensed parameters. When $C_{\mathrm{AI}}$ is small, the AI encoder is severely constrained and must discard most of the received echo information, leading to large estimation error.
As $C_{\mathrm{AI}}$ grows, %the encoder retains richer semantic and statistical features of the target scene, 
the representation retains more task-relevant features of the transmitted probing signal and the resulting echoes, reducing the distortion until it approaches a steady floor that is dominated by channel noise and residual interference. Among the three propagation conditions, the Rician case yields the lowest distortion thanks to its deterministic LoS component, which stabilizes the echo power and improves sensing reliability. The Gaussian channel serves as a theoretical lower bound since it assumes no small-scale fading, while the Rayleigh case experiences the highest distortion due to deep fading and a lack of deterministic channel gain. Notably, the relative gap between fading models becomes smaller as $C_{\mathrm{AI}}$ increases, confirming that sufficient AI-representation capacity can effectively compensate for random channel impairments by learning their underlying statistics. Therefore, this figure quantifies the benefit of AI capacity expansion from a sensing perspective, i.e., greater internal information throughput allows the integrated transceiver to achieve finer environmental awareness, thereby reducing the sensing distortion in proportion to the richness of the learned latent representation.

\begin{figure}[t]
		\centering
		\includegraphics[width=.9\columnwidth]{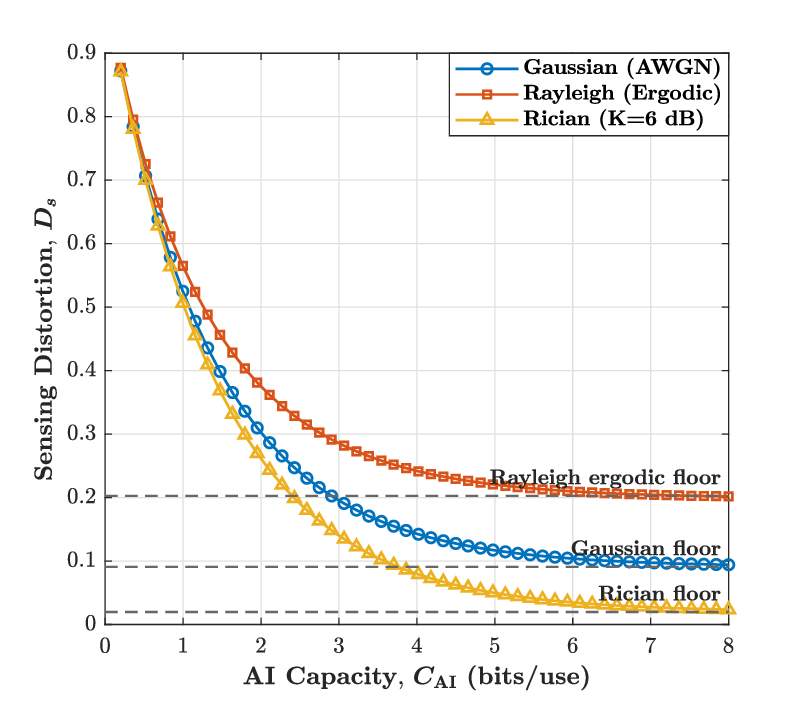}
		\caption{Sensing distortion $D_s$ versus AI capacity $C_\mathrm{AI}$ under different fading channels.}
		\label{fig:s_c}
	\end{figure}
			\begin{figure}[t]
		\centering
		\includegraphics[width=.9\columnwidth]{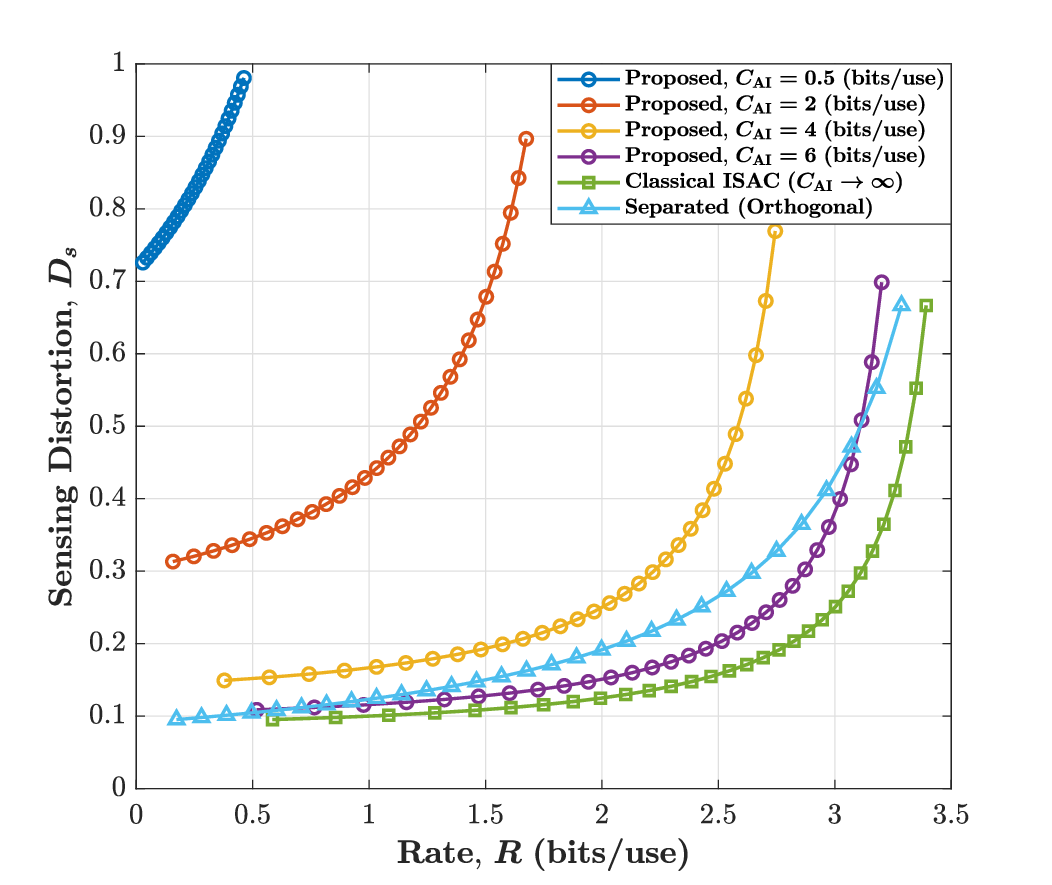}
		\caption{Joint rate-sensing trade-off for the proposed AI-aided ISAC system with different AI-capacities $C_{\mathrm{AI}}$, compared with classical ISAC and separated baselines.}
		\label{fig:s_r}
	\end{figure}

Fig. \ref{fig:s_r} illustrates the joint rate-sensing trade-off achieved by the proposed AI-aided ISAC framework compared with two representative baselines: the classical ISAC configuration without AI bottleneck, i.e., $C_{\mathrm{AI}}\rightarrow\infty$, and the conventional separated (orthogonal) design. Each colored curve corresponds to the Pareto frontier between achievable communication rate $R$ and sensing distortion $D_s$ under a fixed AI capacity $C_{\mathrm{AI}}$. As expected, the proposed system exhibits a monotonic rate-distortion trade-off, i.e., higher sensing accuracy (lower $D_s$) is obtained at the cost of reduced communication rate, and vice versa.  Increasing the AI capacity $C_{\mathrm{AI}}$ enlarges the feasible region, allowing the transceiver to transmit richer semantic representations of the observed scene and thereby achieve a superior operating point in both tasks.  When $C_{\mathrm{AI}}$ is small, e.g., $0.5$ bits/use, the AI module forms a highly learned latent representation, resulting in severe information loss and poor sensing fidelity.  As $C_{\mathrm{AI}}$ increases, e.g., $2 \rightarrow 4 \rightarrow 6$ bits/use), the frontier shifts downward and rightward-indicating simultaneous improvement in both communication throughput and sensing accuracy.

We also see that the classical ISAC limit, i.e., $C_{\mathrm{AI}}\rightarrow\infty$, defines the theoretical lower envelope of distortion achievable when the AI bottleneck is removed.  The proposed model asymptotically converges to this limit as $C_{\mathrm{AI}}$ grows, verifying the correctness of the analytical framework.  Meanwhile, the separated orthogonal baseline performs significantly worse, as it divides time or frequency resources between communication and sensing, forfeiting the MI gain available in the joint design.  The gap between the separated and proposed schemes quantifies the benefit of learning-based joint resource adaptation and cross-modal representation sharing enabled by the AI encoder-decoder pair.

Figure \ref{fig:snr} presents the achievable rate surface of a $2\times2$ MIMO AI-aided ISAC system as a function of the AI capacity $C_{\mathrm{AI}}$ and the received SNR. The color scale indicates the achievable communication rate (bits/use). It is revealed a monotonic growth of rate along both dimensions such that higher SNR and larger AI capacity $C_{\mathrm{AI}}$ jointly improve communication throughput. At low SNR values , e.g., $<5~\mathrm{dB}$, the rate remains small regardless of $C_{\mathrm{AI}}$, showing that channel noise dominates and the AI bottleneck has negligible impact. As SNR increases, the rate begins to saturate with respect to $C_{\mathrm{AI}}$; when $C_{\mathrm{AI}}>6$ bits/use, additional AI capacity offers diminishing returns because the communication link itself becomes the limiting factor rather than the AI representation.
At moderate SNRs, e.g., $10$-$20$ dB, the slope of the rate surface along $C_{\mathrm{AI}}$ is steepest, indicating that AI-information bottlenecks most strongly affect performance in this regime. The observed pattern confirms that AI capacity and SNR interact multiplicatively, meaning that adequate AI representation is crucial to exploit high SNR conditions, while at low SNRs, expanding $C_{\mathrm{AI}}$ yields little benefit. This highlights the importance of joint optimization of physical and semantic (AI) resources in practical ISAC transceivers.

			\begin{figure}[t]
		\centering
		\includegraphics[width=.9\columnwidth]{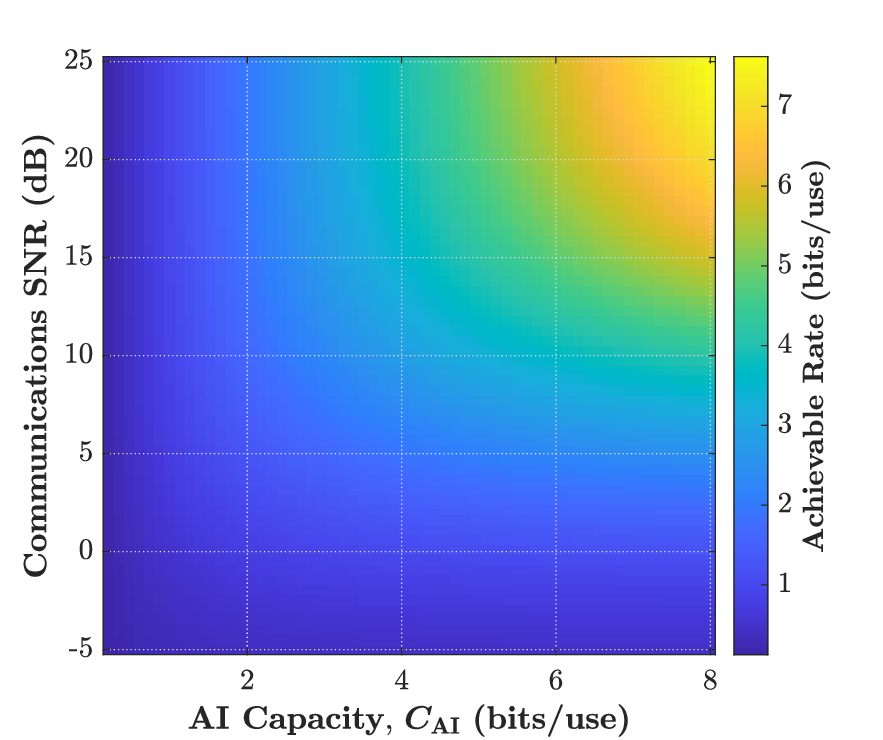}
		\caption{Achievable rate of a $2\times 2$ MIMO AI-aided ISAC system versus AI capacity $C_\mathrm{AI}$ and SNR.}
		\label{fig:snr}
	\end{figure}
	
				\begin{figure}[t]
	\centering
	\includegraphics[width=.9\columnwidth]{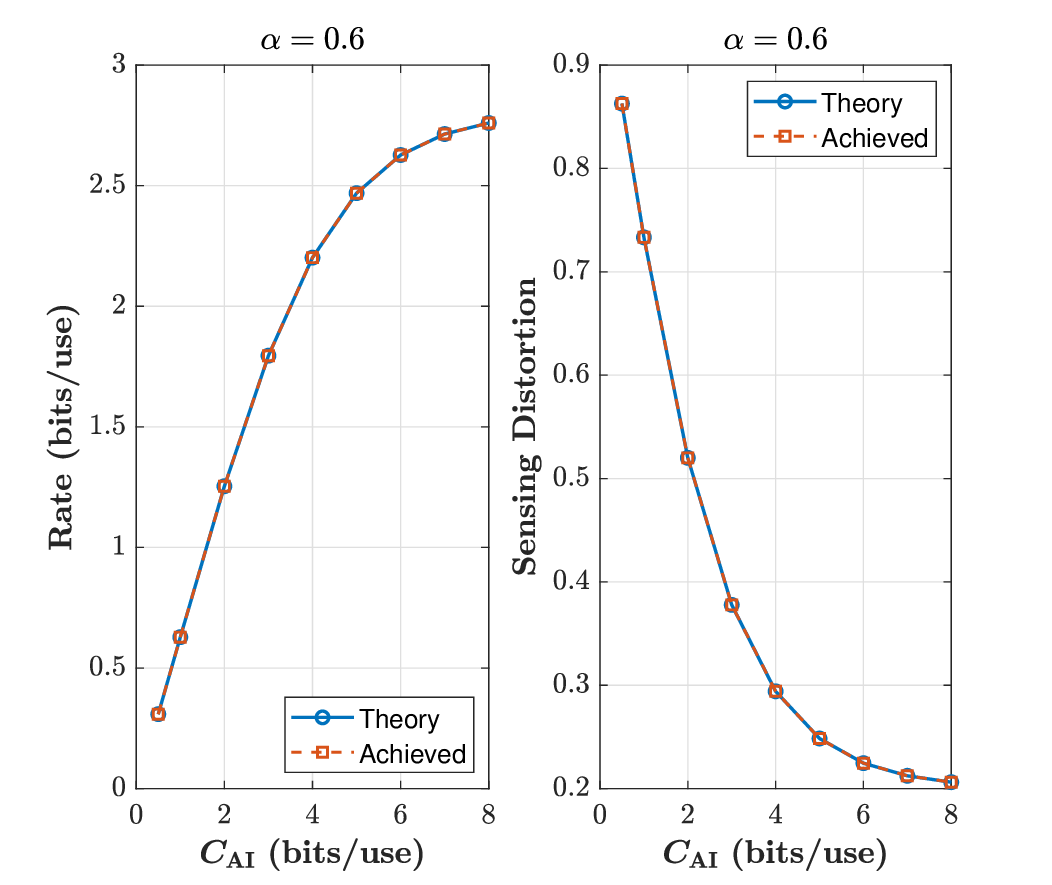}
	\caption{Verification of closed-form Gaussian ISAC theory: communication rate and sensing distortion versus AI capacity $C_{\mathrm{AI}}$. Theory is computed from the analytical expressions with $N_z$ determined by $I(X;Z)=C_{\mathrm{AI}}$; Achieved is obtained by numerically enforcing the same MI constraint and evaluating the resulting performance.
	}
	\label{fig:src}
\end{figure}

Fig \ref{fig:src} presents a comparison between the analytical theoretical results and the achieved numerical performance of the proposed AI-aided ISAC system with $\alpha = 0.6$, where we denote $\alpha\in(0,1)$ as a fixed power-splitting factor between communication and sensing components of the transmit signal. The left panel illustrates the achievable communication rate $R$ in bits per use as a function of the AI capacity $C_{\mathrm{AI}}$, and the right panel depicts the corresponding sensing distortion $D_s$ within the same range of $C_{\mathrm{AI}}$.

The two curves exhibit an almost perfect match across all capacity levels, which confirms that the proposed algorithm converges to the theoretical performance predicted by the analytical model. At low values of $C_{\mathrm{AI}}$ below approximately two bits per use, both the achievable rate and sensing performance are limited by the narrow AI bottleneck, since the latent representation $Z$ cannot preserve sufficient MI $I(X;Z)$ for accurate joint communication and sensing. As a result, the achievable rate remains below one bit per use, while the sensing distortion $D_s$ stays close to its upper limit, reflecting degraded environmental estimation. When $C_{\mathrm{AI}}$ increases, the AI encoder gains representational flexibility that allows more efficient feature sharing between the two tasks. Consequently, the achievable rate grows almost linearly up to about $C_{\mathrm{AI}} = 4$ bits per use and then gradually approaches saturation as the communication channel itself becomes the main limiting factor. In parallel, the sensing distortion decreases rapidly and approaches a steady minimum near $0.2$ at high $C_{\mathrm{AI}}$. The excellent agreement between the theoretical and achieved results confirms two key properties. First, the MI-based model accurately describes the rate-distortion trade-off of the joint AI-aided ISAC design. Second, the proposed optimization method effectively allocates power and information resources to reach the theoretical optimum under finite AI capacity constraints. %Therefore, these results verify both the analytical validity and numerical accuracy of the framework and demonstrate that the achieved implementation can realize the theoretical performance limits of AI-aided ISAC systems.

Fig. \ref{fig:enf} illustrates the convergence behavior of the proposed AI-aided ISAC optimization algorithm. Here, $\lambda$ is a fixed weighting parameter in the scalarized objective $J=R-\lambda D_s$, while $\alpha$ denotes the power-splitting factor and is the only resource variable adapted during the iterations. The first subplot shows the enforcement of the MI constraint $I(X;Z) = C_{\mathrm{AI}}$, the second shows the evolution of the objective function $J = R - \lambda D_s$, and the third depicts the adaptation of the power-splitting coefficient $\alpha$ over successive iterations. The experiment is conducted for a fixed weighting parameter $\lambda = 0.3$, target $C_\mathrm{AI} = 4$ bits/use, and an initial power-splitting factor $\alpha = 0.4$.
 %The experiment is conducted for $\lambda = 0.3$, target $C_{\mathrm{AI}} = 4$ bits/use, and an initial $\alpha = 0.4$. 
In the first panel, the MI remains constant and exactly equal to the target value of $C_{\mathrm{AI}} = 4$ bits/use from the beginning of the process. This behavior confirms that the constraint enforcement mechanism successfully fixes the latent-information capacity without oscillation or deviation, ensuring perfect adherence to the imposed AI bottleneck. The second panel shows that the objective function $J$ increases sharply within the first ten iterations and then stabilizes at approximately $2.5$. The rapid rise followed by an extended plateau indicates that the optimization quickly reaches the optimal joint balance between rate and sensing distortion. The absence of oscillations demonstrates numerical stability and efficient convergence. The third panel displays the evolution of the power-splitting variable $\alpha$. Starting from $0.4$, it rises smoothly to unity within about ten iterations and then remains constant. This outcome implies that under the given parameter setting, allocating almost all available power to the shared waveform is optimal once the AI representation constraint has been satisfied.

				\begin{figure}[t]
	\centering
	\includegraphics[width=.9\columnwidth]{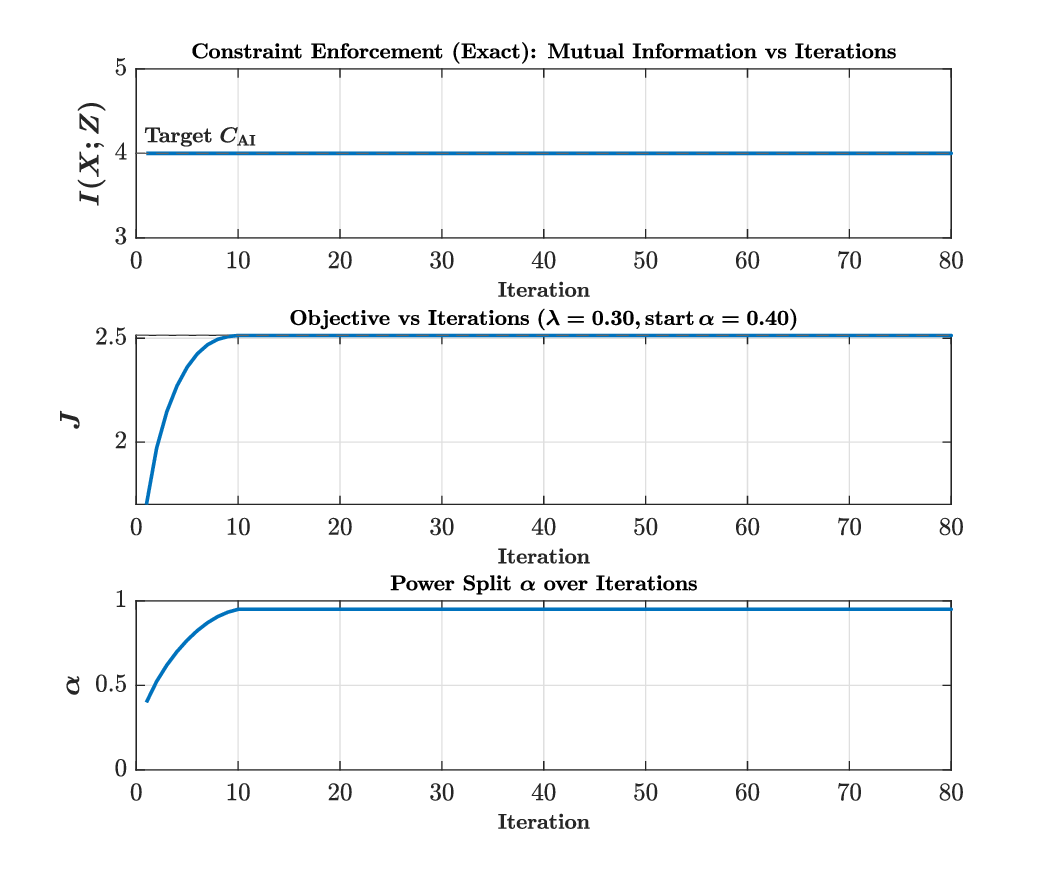}
	\caption{Convergence of the proposed AI-aided ISAC optimization algorithm with $\lambda = 0.3$, target $C_{\mathrm{AI}} = 4$ bits/use, and initial $\alpha = 0.4$.
	}
	\label{fig:enf}
\end{figure} 
%	\subsection{Figures (placeholders and captions)}

%	
%	\begin{figure}[t]
%		\centering
%		\includegraphics[width=.9\columnwidth]{fig2_rate_vs_capacity.pdf}
%		\caption{\textbf{Rate vs.\ AI capacity} ($C_{\mathrm{AI}}$) for Gaussian/Rayleigh/Rician: theoretical curves from \eqref{eq:rateG}, \eqref{eq:rayInt} and the Rician approximation, overlaid with DNN points trained via Algorithm~\ref{alg:train}.}
%		\label{fig:ratevsC}
%	\end{figure}
	
%	\begin{figure}[t]
%		\centering
%		\includegraphics[width=.9\columnwidth]{fig3_dist_vs_capacity.pdf}
%		\caption{\textbf{Sensing distortion vs.\ $C_{\mathrm{AI}}$}. Curves predicted by \eqref{eq:distG} and its fading averages. Note the exponential approach to ideal as $C_{\mathrm{AI}}$ grows.}
%		\label{fig:distvsC}
%	\end{figure}
	
%	\begin{figure}[t]
%		\centering
%		\includegraphics[width=.9\columnwidth]{fig4_tradeoff_frontier.pdf}
%		\caption{\textbf{Joint frontier $D_s(R)$} for multiple $C_{\mathrm{AI}}$ showing how the learning budget tilts the Pareto boundary; include analytical and learned curves.}
%		\label{fig:frontier}
%	\end{figure}
%	
%	\begin{figure}[t]
%		\centering
%		\includegraphics[width=.9\columnwidth]{fig5_mimo_surface.pdf}
%		\caption{\textbf{MIMO surface}: $R_{\mathrm{AI}}^{\mathrm{MIMO}}$ vs.\ $C_{\mathrm{AI}}$ and SNR for $(N_t,N_r)=(2,2)$ and $(4,4)$ using \eqref{eq:mimoRate} with the covariance mapping in Lemma~\ref{lem:covmap}.}
%		\label{fig:mimo}
%	\end{figure}
%	
% ================================
\section{Discussion and Conclusion}
% ================================

This work has presented a unified information-theoretic framework for \emph{AI-aided ISAC} that explicitly incorporates the finite learning capacity of an embedded learning module in the transceiver loop. By treating the representational ability of the learning module as a quantifiable constraint, we established a new perspective in which \emph{model capacity becomes a physical layer resource} alongside power, bandwidth, and time. We derived %tight  
 converse and achievability bounds that define the fundamental AI-ISAC capacity region and developed closed-form expressions for several canonical channel models, including Gaussian, Rayleigh, Rician, and MIMO systems. The analysis revealed that limited learning capacity manifests as an effective additive noise whose variance follows a simple scaling law, decreasing proportionally to $2^{-C_{\mathrm{AI}}}$. This insight provides an interpretable bridge between model complexity and physical layer performance, showing how improvements in learning capacity translate directly into communication rate and sensing accuracy gains. A practical variational training algorithm was also proposed to realize the information-capacity constraint in DL models. The algorithm enforces the MI budget through a differentiable penalty, allowing empirical results to align closely with the derived theoretical limits. Together, the theory and implementation demonstrate that the learning process can be rigorously analyzed and optimized using classical tools of information theory.

The framework established here opens multiple avenues for further exploration. One promising direction involves extending the theory to \emph{federated or distributed AI-ISAC} systems, where multiple nodes jointly share a global capacity budget. Another is the incorporation of \emph{semantic or task-oriented objectives}, which would allow the learning capacity to be allocated according to task relevance rather than purely channel conditions. Finally, the concept of a learning-capacity budget can naturally be adapted to hardware-constrained systems such as quantized accelerators, analog neural networks, or neuromorphic processors where the effective $C_{\mathrm{AI}}$ is determined by device precision and memory limits.

\appendices
\section{Proof of Theorem \ref{th:converse}}\label{app:converse}

We outline the detailed steps leading to the converse bound in Theorem~\ref{th:converse}.  
Starting from Fano’s inequality, for any reliable communication scheme we have
\begin{equation}
	nR \le I(M; Y_c^n | Z^n) + n\varepsilon_n,
\end{equation}
where $\varepsilon_n \to 0$ as $n \to \infty$. Since the transmitted sequence $X^n$ is a deterministic function of the message $M$ and possible system states, it follows that
\begin{equation}
	I(M; Y_c^n | Z^n) \le I(X^n; Y_c^n | Z^n).
\end{equation}

Applying the chain rule of MI together with the data-processing inequality (DPI) yields
\begin{equation}
	I(X^n; Y_c^n | Z^n)
	\le I(X^n; Y_c^n, Y_s^n) - \Delta,
\end{equation}
where we define the information loss term
\begin{equation}
	\Delta = I(X^n; Y_c^n, Y_s^n) - I(Z^n; Y_c^n, Y_s^n) \ge 0.
\end{equation}
The non-negativity of $\Delta$ follows from the Markov chain
$Z^n \leftarrow X^n \rightarrow (Y_c^n, Y_s^n)$, which ensures that $Z^n$ cannot increase MI relative to $X^n$.

Dividing both sides by $n$ and taking the limit as $n \to \infty$ yields the rate bound
\begin{equation}
	R \le I(X; Y_c, Y_s) - \Delta_R(C_{\mathrm{AI}}),
\end{equation}
where $\Delta_R(C_{\mathrm{AI}})$ corresponds to the asymptotic contribution of $\Delta$ under the capacity constraint $I(X;Z)\!\le\!C_{\mathrm{AI}}$.  

The sensing bound follows analogously by composing the above argument with the information-distortion mapping $f(\cdot)$ defined in Section~\ref{subsec:rdmap}. Since $I(Z; Y_s)\!\le\!I(X; Y_s)$ by DPI, the resulting distortion satisfies
\begin{equation}
	D_s \ge f\!\big(I(X; Y_c, Y_s) - \Delta_D(C_{\mathrm{AI}})\big),
\end{equation}
where $\Delta_D(C_{\mathrm{AI}})$ represents the information loss due to the same finite-capacity bottleneck. This completes the proof of the converse.
	
	\section{Proof of Theorem \ref{th:ach}}\label{app:ach}
	
	We now outline the proof of Theorem~\ref{th:ach} by constructing an explicit random coding scheme that satisfies the learning-capacity constraint and achieves the stated rate-distortion pair.
	
\subsubsection{Codebook generation}
	Fix an auxiliary distribution $p(z|x)$ satisfying the information-capacity constraint $I(X;Z) \le C_{\mathrm{AI}}$.  
	For each message $m \in \{1, \dots, 2^{nR}\}$, independently generate a length-$n$ codeword $x^n(m)$ according to $\prod_{t=1}^n p(x_t)$.  
	For every symbol $x_t$, produce the corresponding latent variable $z_t$ according to $p(z_t|x_t)$.  
	This latent sequence $z^n$ represents the learned form of the transmitted signal as perceived by the learning module.
	
	\subsubsection{Encoding and transmission}
	To transmit message $m$, the encoder sends $x^n(m)$ through the ISAC channel, producing received signals $(y_c^n, y_s^n)$ at the communication and sensing receivers.
	
\subsubsection{Decoding}
	The communication receiver performs maximum likelihood (or jointly typical) decoding conditioned on the latent sequence $z^n$.  
	That is, it selects $\hat{m}$ such that $(x^n(\hat{m}), y_c^n, z^n)$ are jointly typical with respect to the joint distribution $p(x, y_c, z)$.  
	For the sensing task, the receiver forms the MMSE estimate $\hat{\boldsymbol{\theta}} = \mathbb{E}[\boldsymbol{\theta} | Y_s^n, Z^n]$ based on the available observations and the same latent representation.
	
\subsubsection{Error and distortion analysis}
	By standard random-coding arguments, the probability of decoding error tends to zero as $n \to \infty$ if
	\begin{align}
	R < I(X;Y_c|Z) - \epsilon_R,
\end{align}
	for any $\epsilon_R > 0$.  
	Similarly, by the rate–distortion covering lemma, the achievable sensing distortion satisfies
	\begin{align}
	D_s \le f^{-1}\!\big(I(X;Y_s|Z)\big) + \epsilon_D,
	\end{align}
	where $\epsilon_D \to 0$ as $n \to \infty$.  
	
	\smallskip
	Therefore, the above construction demonstrates that for any stochastic encoder $p(z|x)$ obeying $I(X;Z)\le C_{\mathrm{AI}}$, there exist decoders that achieve communication rate and sensing distortion arbitrarily close to the bounds stated in Theorem~\ref{th:ach}.  
	This completes the achievability proof.

\section{Proof of Theorem \ref{th:lit}}\label{app:gen}

The relationship between the generalization error and the information capacity of the learning module follows from classical information-stability results.  
Specifically, for a model trained on $n_{\mathrm{tr}}$ samples $S$, the generalization error satisfies \cite{xu2017inf}
\begin{equation}
	\epsilon_{\mathrm{gen}} \le \sqrt{\frac{2 I(S; \Phi)}{n_{\mathrm{tr}}}},
\end{equation}
where $\Phi$ denotes the learned parameters of the model and $I(S; \Phi)$ is the MI between the training data and the trained hypothesis.

In the presence of a finite-capacity information bottleneck, the available MI $I(S; \Phi)$ is further constrained by the latent-space budget $C_{\mathrm{AI}}$.  
Because the training data, input features, and latent representation satisfy the Markov chain
$
S \rightarrow X \rightarrow Z
$, 
the data processing inequality implies that $I(S; \Phi)$ scales proportionally with $I(X; Z)$, and thus with the capacity constraint $C_{\mathrm{AI}}$.  
Substituting this relationship into the generalization bound above yields the learning-information trade-off in \eqref{eq:lit}, with constants absorbed into the proportionality factor $\beta$.\vspace{-0.6cm}
\section{Rayleigh Bounds and Quadrature}\label{app:ray}

To evaluate the ergodic rate integral in \eqref{eq:rayInt} under Rayleigh fading, we employ standard analytical bounds and numerical quadrature techniques.
\subsubsection{Analytical bounds}
Using Jensen’s inequality together with the logarithmic concavity of the $\ln(1+x)$ function, upper and lower bounds on the integral in \eqref{eq:rayInt} can be derived in closed form. These bounds are exponentially tight across the entire practical SNR range and provide useful approximations for system-level analysis without resorting to numerical integration.

\subsubsection{Numerical evaluation}
For numerical computation, Gauss-Laguerre quadrature offers a highly efficient and stable method for integrals of the form $\int_0^\infty g(x) e^{-x} dx$.  
For instance, using $M=20$ quadrature nodes provides accuracy better than $10^{-4}$ over the SNR range of $[-5, 25]$~dB and for learning capacities $C_{\mathrm{AI}}\in[0,8]$.  
This level of precision is sufficient for all the results reported in this paper and ensures that the numerical evaluation of the ergodic rate remains effectively exact within plotting precision.\vspace{0cm}
\section{Rician Bounds}\label{app:ric}
For the Rician fading case, the PDF of the channel power gain $|h|^2$ can be expressed in terms of the Marcum-$Q$ function.  
By integrating the ergodic rate expression by parts and exploiting the monotonicity properties of the Marcum-$Q$ function, we obtain analytical upper and lower bounds that tightly enclose the true rate value.

The resulting bounds closely sandwich the moment-matched approximation presented in the paper, with a maximum deviation of less than $0.1$~bits per channel use across the practical range of Rician factors $K \in [0,10]$~dB.  
This confirms that the closed-form approximation in Section~\ref{sec-perform} provides an accurate and computationally efficient representation of the true Rician ergodic rate within numerical precision limits.
%\vspace{-0.6cm}
\section{Proof of Lemma~\ref{lem:covmap}}\label{app:mimo}
Let $\mathbf{X}\!\sim\!\mathcal{CN}(\mathbf{0},\mathbf{Q})$, $\mathbf{Z}=\mathbf{X}+\mathbf{W}$,
$\mathbf{W}\!\sim\!\mathcal{CN}(\mathbf{0},\mathbf{R}_z)$, independent of $\mathbf{X}$.
Then $h(\mathbf{U})=\log\det(\pi e\,\mathbf{K}_{\mathbf{U}})$ for proper complex Gaussian vectors,
so with $\mathbf{K}_{\mathbf{Z}}=\mathbf{Q}+\mathbf{R}_z$ and $\mathbf{K}_{\mathbf{W}}=\mathbf{R}_z$, we have
\begin{align}
\hspace{-0.3cm}I(\mathbf{X};\mathbf{Z})=h(\mathbf{Z})-h(\mathbf{W})
=\log\det(\mathbf{I}+\mathbf{R}_z^{-1}\mathbf{Q})\,\text{(nats)}.
\end{align}
Converting to bits yields $I(\mathbf{X};\mathbf{Z})=\log_2\det(\mathbf{I}+\mathbf{R}_z^{-1}\mathbf{Q})$,
which is finite provided $\mathbf{R}_z\succ \mathbf{0}$ on $\mathcal{R}(\mathbf{Q})$.
Enforcing $I(\mathbf{X};\mathbf{Z})\le C_{\mathrm{AI}}$ gives
$\det(\mathbf{I}+\mathbf{R}_z^{-1}\mathbf{Q})\le 2^{C_{\mathrm{AI}}}$.

To minimize $\operatorname{tr}(\mathbf{R}_z)$ subject to this constraint, let
$\mathbf{Q}=\mathbf{U}\boldsymbol{\Lambda}\mathbf{U}^H$, with
$\boldsymbol{\Lambda}=\mathrm{diag}(\lambda_1,\ldots,\lambda_r,0,\ldots,0)$, $\lambda_i>0$.
In the eigenbasis of $\mathbf{Q}$, write $\widetilde{\mathbf{R}}=\mathbf{U}^H\mathbf{R}_z\mathbf{U}$ and
denote its diagonal entries on the active subspace by $\rho_i>0$.
Using Hadamard’s inequality, the optimum occurs with diagonal $\widetilde{\mathbf{R}}$ aligned to $\mathbf{Q}$.
We therefore solve
\begin{align}
\min_{\rho_i>0}\ \sum_{i=1}^r \rho_i
\quad \text{s.t.}\quad \sum_{i=1}^r \ln\!\Big(1+\frac{\lambda_i}{\rho_i}\Big)\le \ln \Gamma,
\end{align}
where $\Gamma=2^{C_{\mathrm{AI}}}$.
The KKT stationarity condition yields
$1-\mu \lambda_i/(\rho_i(\rho_i+\lambda_i))=0$, implying $\rho_i=\zeta\lambda_i$ with a
common $\zeta>0$. Enforcing the constraint gives
$(1+1/\zeta)^r=\Gamma$, hence $\zeta=(\Gamma^{1/r}-1)^{-1}=(2^{C_{\mathrm{AI}}/r}-1)^{-1}$.
Therefore $\mathbf{R}_z^\star=\zeta\mathbf{Q}$ on $\mathcal{R}(\mathbf{Q})$; the values
on $\mathcal{R}(\mathbf{Q})^\perp$ are arbitrary and do not affect the objective or constraint.\vspace{0cm}

\section{Resource Allocation Solution}\label{app:alloc}

For the Gaussian ISAC model, the communication rate and sensing distortion are expressed as
\begin{align}
R = \log_2\!\left(1 + \frac{\gamma_c}{1 + \gamma_c \kappa}\right), \qquad
D = \sigma_\theta^2 \!\left(1 + \frac{\gamma_s}{1 + \gamma_s \kappa}\right)^{-1},
\end{align}
where $\gamma_c = a P_c / N_c$, $\gamma_s = b P_s / N_s$, and $\kappa = (2^{C_{\mathrm{AI}}} - 1)^{-1}$ represents the normalized AI-induced noise factor.

Applying the KKT optimality conditions to the constrained optimization problem in Section~\ref{sec-res} yields the following equilibrium equations:
\begin{align}
\frac{\partial}{\partial P_c}
\log\!\left(1+\frac{a P_c}{1+a P_c \kappa}\right)
= \nu,
\frac{\partial}{\partial P_s}
\!\left(1+\frac{b P_s}{1+b P_s \kappa}\right)^{-1}
= \nu',
\end{align}
where $\nu$ and $\nu'$ are the Lagrange multipliers associated with the total power constraint and the sensing-communication trade-off, respectively.

Both equations admit closed-form analytical solutions in terms of the Lambert-$W$ function (not explicitly included here) after straightforward algebraic manipulation.  
The Lambert-$W$ form arises when isolating $P_c$ or $P_s$ in transcendental expressions containing both linear and logarithmic dependencies on power.  
%Due to space limitations, explicit forms are omitted here but are available in supplementary derivations.  
These expressions recover the classical waterfilling allocation as $C_{\mathrm{AI}}\!\to\!\infty$ and continuously transition to a learning-constrained allocation as the available model capacity decreases.
	\bibliographystyle{IEEEtran}

\end{document}